%% file: CIKM.tex
\documentclass[nonacm=true, sigconf, review=false]{acmart}
\usepackage{subfigure}
\usepackage{lineno}
\usepackage[table]{colortbl}
\usepackage{pifont}
\usepackage{algpseudocode}
\usepackage{rotating}
\usepackage[T1]{fontenc}
\usepackage[export]{adjustbox}

\usepackage{graphicx, subfigure, multirow} 
\include{preamble}

\newcommand{\rtr}{regularly triangle-rich}
\newcommand{\mainproc}{RTRExtractor}
\usepackage{array}
\newcolumntype{H}{>{\setbox0=\hbox\bgroup}c<{\egroup}@{}}

\title{Covering a Graph with Dense Subgraph Families, via Triangle-Rich Sets}
\author{Sabyasachi Basu$^*$}
\email{sbasu3@ucsc.edu}
\affiliation{%
  \institution{University of California Santa Cruz}
  \city{Santa Cruz, CA, USA}
}

\author{Daniel Paul-Pena$^{*\dagger}$}
\email{dpaulpen@ucsc.edu}
\affiliation{%
  \institution{University of California Santa Cruz}
  \city{Santa Cruz, CA, USA}
  }

\author{Kun Qian}
\email{qianku@amazon.com}
\affiliation{%
  \institution{Amazon}
  \city{Palo Alto, CA, USA}
}

\author{C. Seshadhri}
\email{sesh@ucsc.edu}
\affiliation{%
  \institution{University of California Santa Cruz}
  \city{Santa Cruz, CA, USA}
  }
\affiliation{%
  \institution{Amazon}
  \city{Palo Alto, CA, USA}
}

\author{Edward W Huang}
\email{ewhuang@amazon.com}
\affiliation{%
  \institution{Amazon}
  \city{Palo Alto, CA, USA}
}

\author{Karthik Subbian}
\email{ksubbian@amazon.com}
\affiliation{%
  \institution{Amazon}
  \city{Palo Alto, CA, USA}
}

\date{}
\begin{abstract}
Graphs are a fundamental data structure used to represent relationships 
in domains as diverse as the social sciences, bioinformatics, cybersecurity, the Internet, and more.
One of the central observations in network science is that real-world graphs are globally sparse,
yet contains numerous ``pockets" of high edge density. A fundamental task in graph mining is to discover
these dense subgraphs. Most common formulations of the problem involve finding a single (or a few)
``optimally" dense subsets. But in most real applications, one does not care for the optimality.
Instead, we want to find a large collection of dense subsets that covers a significant fraction
of the input graph. 

We give a mathematical formulation of this problem, using a new definition of \emph{regularly triangle-rich (RTR) families}.
These families capture the notion of dense subgraphs that contain many triangles and have degrees
comparable to the subgraph size. We design a provable algorithm, \mainproc, that can
discover RTR families that approximately cover \emph{any RTR} set. The algorithm is efficient and is inspired by recent results that use
triangle counts for community testing and clustering. 

We show that \mainproc{} has excellent behavior on a large variety of real-world datasets. It is able
to process graphs with hundreds of millions of edges within minutes. Across many datasets, \mainproc{} achieves
high coverage using high edge density datasets. For example, the output covers a quarter
of the vertices with subgraphs of edge density more than (say) $0.5$, for datasets with 10M+ edges.
We show
an example of how the output of \mainproc{} correlates with meaningful sets of similar vertices in a citation network, demonstrating the utility of \mainproc{} for unsupervised graph discovery tasks.

\end{abstract}

\begin{document}
\captionsetup[subfigure]{labelformat=empty}
\settopmatter{printfolios=true}
\maketitle
\def\thefootnote{*}\footnotetext{These authors contributed equally to this work.}
\def\thefootnote{$\dagger$}\footnotetext{Work done while author was an intern at Amazon.}

\section{Introduction} \label{sec:intro}
\begin{figure*}
\centering
\begin{subfigure}
         {\includegraphics[width= 0.16\textwidth]{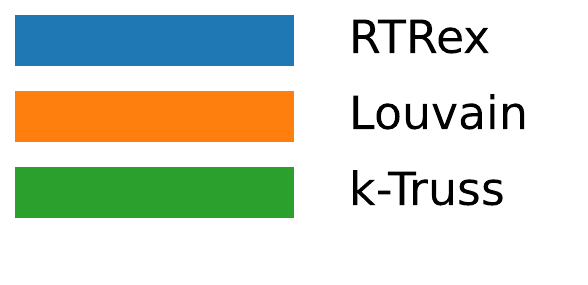}}
     \end{subfigure}
     \begin{subfigure}{}
         \begin{turn}{90}
            \hspace{1.8cm}\scriptsize{Coverage}             
         \end{turn}
     \end{subfigure}
        \begin{subfigure}
         {\includegraphics[width= 0.18\textwidth]{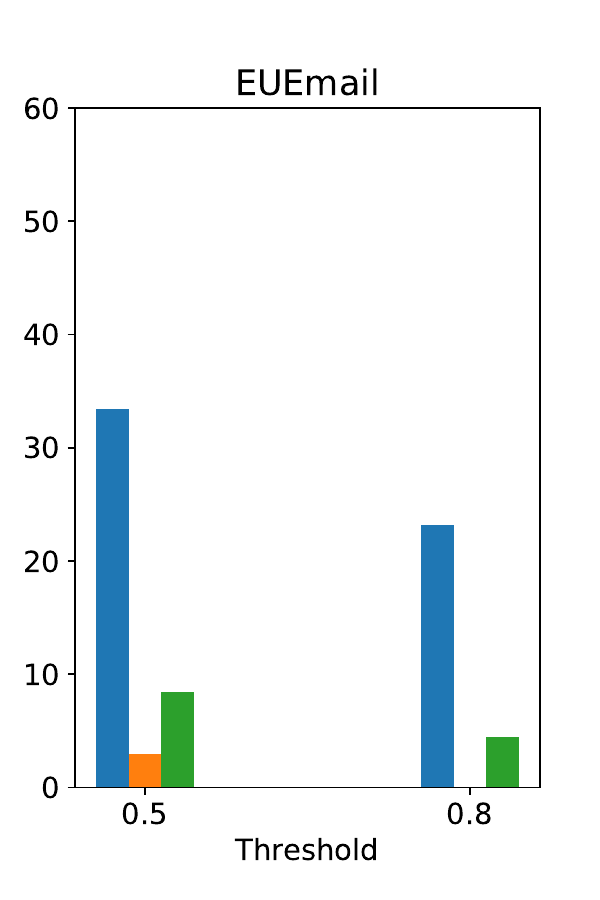}}
     \end{subfigure}
     \begin{subfigure}
         {\includegraphics[width= 0.18\textwidth]{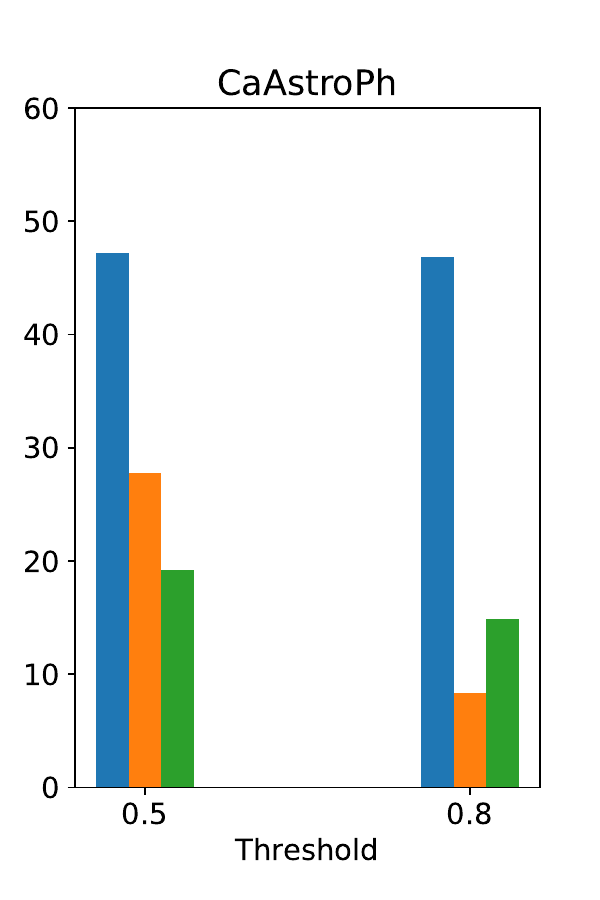}}
     \end{subfigure}
     \begin{subfigure}
         {\includegraphics[width= 0.18\textwidth]{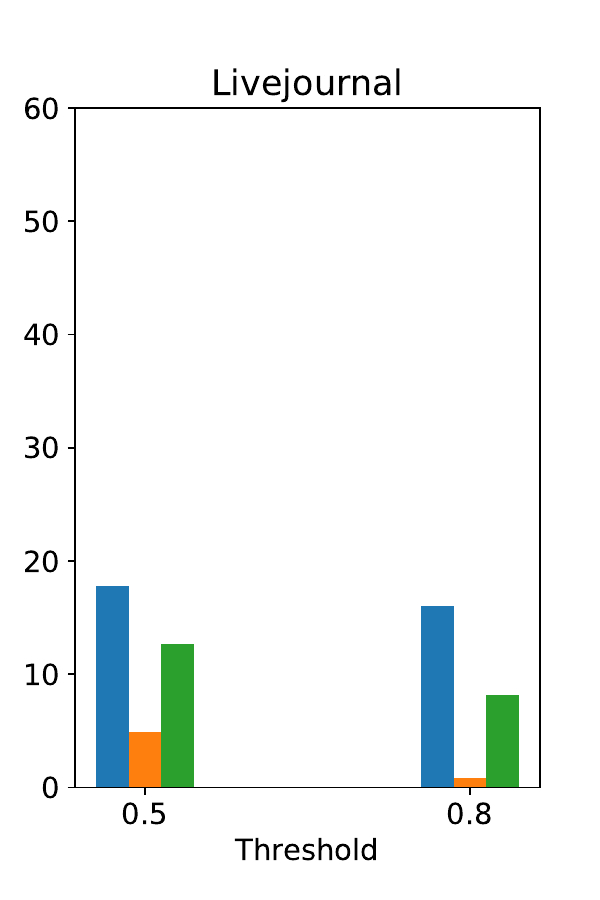}}
     \end{subfigure}
     \begin{subfigure}
         {\includegraphics[width= 0.18\textwidth]{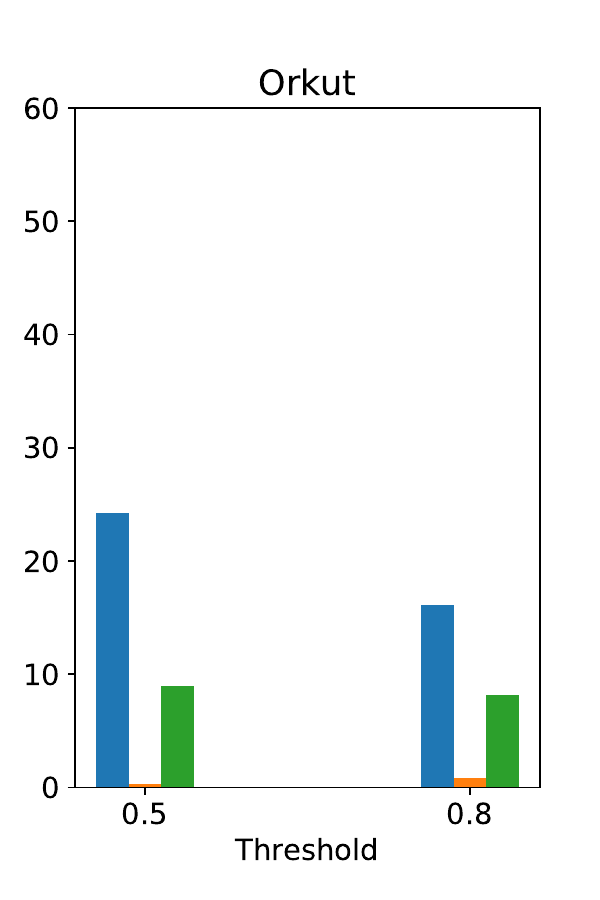}}
     \end{subfigure}
     
  \caption{Coverage of \mainproc{} compared to the two best performing competitors in a variety of networks of different sizes: from thousands to millions of vertices. For each method, we compute the fraction of vertices in sets of 5 vertices or more, and of density more than 0.5 and 0.8.
  \mainproc{} consistently covers significantly higher percentages in  dense clusters.}
  \vspace{-.5cm}
  \label{fig:intro}
\end{figure*}
Graphs are a fundamental data structure used to represent relationships 
in domains as diverse as the social sciences, bioinformatics, cybersecurity, the Internet, and more.
One of the central observations in network science is that real-world graphs are globally sparse,
yet contains numerous ``pockets" of density. A fundamental task in graph mining is to discover
these dense subgraphs~\cite{dense}.

Dense subgraph discovery has a rich set of applications. 
It has been used for finding cohesive social groups~\cite{Beal2003CohesionAP,forsyth2010group}, spam link farms in web graphs~\cite{KRRT99, DGP07, GKT05}, graph visualization~\cite{AhDBV05}, real-time story identification~\cite{ASKS12}, 
DNA motif detection in biological networks~\cite{FNBB06}, finding correlated genes~\cite{ZH05},
epilepsy prediction~\cite{Iasemidis03}, finding price value motifs in financial data~\cite{DJDLT09}, graph compression~\cite{BC08}, distance query indexing~\cite{JXRF09}, and 
increasing the throughput of social networking site servers~\cite{GJLSW13}. 
Dense subgraph discovery is a typical unsupervised machine learning task and a central
tool for knowledge/structure discovery in large networks.

\begin{figure}
    \centering
    \includegraphics[width = \linewidth]{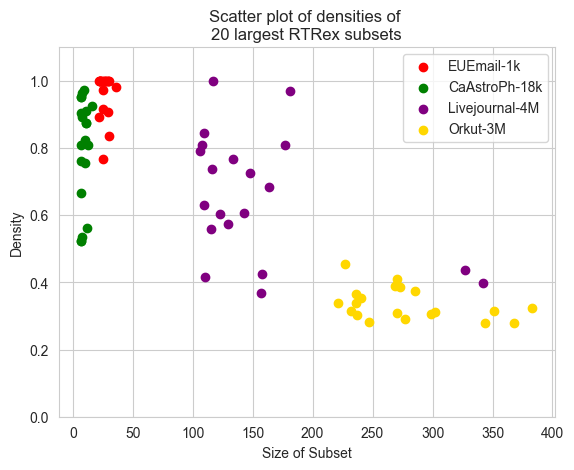}
    \caption{The largest 20 \mainproc{} subsets produced on some datasets. and their respective densities. We note that even for the largest subsets, which in some cases may have hundreds of vertices, density is still remarkably high.}
    \label{fig:largest-scatter}
    \vspace{-.5cm}
\end{figure}

There are numerous challenges in designing algorithms
in the context of the above applications. Consider the input graph $G = (V,E)$.
The \emph{density} of a subset $S \subset V$ is $E(S)/{|S| \choose 2}$,
where $E(S)$ is the number of edges inside $S$. The density has a maximum value
of $1$, when $S$ is a clique. The typical density of a real-world (say) social network
with millions of vertices is less than $10^{-5}$, but there are often sets of size (say)
$20$ with density above $0.5$. Our aim is to find sufficiently large sets that are sufficiently dense (in practice, tens to few hundreds of vertices and density $0.3$ or above).

Most of the algorithms and data mining literature phrase this as an optimization
problem~\cite{Gol84}, either finding the largest set above a given density, or finding the densest 
set above or below a given size~\cite{AndC09}. Another popular variant is related to the notion of correlation clustering~\cite{BBC04}, which has received a lot of theoretical attention in recent years. 
Typical formulations are NP-hard~\cite{Håstad1999, Fe02, Khot06}, and most practical
procedures use heuristics and approximation algorithms~\cite{Charikar, AndC09, Flowless, Ts15, Tsou14, KS22, DCS17, CQT22}. In the past decade, the problem of finding one or many dense subgraphs has been approached in various forms by the data mining community~\cite{BGMT21, Sot20, GKT05, SEF16}; we refer the reader to the recent survey by Lanciano et al. for a more comprehensive list~\cite{LMFB23}.
From the viewpoint of applications and as an unsupervised ML technique, this formulation
has shortcomings for real-world networks. The vanilla densest-subgraph variant is prone to extracting massive subgraphs, the work on $k$-way cuts asks for a specific, prescribed $k$ as the number of partitions, and $k$-densest subgraphs ask for a lower bound on the size of each partition (and the best known approximation factors are quite poor). One rarely cares about the exact (or even approximate) optimum. Since
data has noise, it is possible that a worse solution is still meaningful. Moreover,
applications require finding \emph{many} possible solutions. From a knowledge discovery
perspective, if a dense set of vertices indicates special structure, we would like
to find many such sets. Even if we do not capture the single optimal dense set exactly,
we would like to learn about as many dense sets as possible. Often each dense set represents some "information" or structure to be discovered, so we want a large disjoint collection of such sets. Dense subgraphs~\cite{SubGNN, FBBM22} are used as primitives for downstream ML-tasks, so the ability to cover more vertices in dense subgraphs would be an obvious benefit.

Our goal for this paper is the following. \emph{Can we design algorithms that cover a significant portion of the graph, using disjoint large, high density subgraphs?} From an empirical standpoint, we want
an algorithm that can scale to graphs with hundreds of millions of edges, and can find 
many dense subgraphs in real data.
Towards this goal, we would like a theoretical framework that leads to practical and provable algorithms. \emph{Can we give a formal setup, where \emph{many} dense subgraphs can be found tractably?}

\subsection{Our Contributions} \label{sec:contrib}

We give a new theoretical setup for dense subgraph discovery, accounting for triangle density. We design
an algorithm, \mainproc, that is a provable (approximation) algorithm for this version of dense subgraph
discovery. We show that an implementation of \mainproc{} is extremely successful at finding many dense subgraphs
in a variety of real-world datasets.
\begin{table*}[]
      \centering
      \begin{tabular}{|c||c|c|c|c|c|c|c|c|c|c|c|c|} 
      \hline 
      $|T|\geq$ & EUEmail &Hamsterster & CaAstroPh & Epinions & DBLP &Youtube & Skitter & Wiki & LJ & Orkut \\ \hline
10 &0.82 &0.97 &0.99 &0.87 &1.0 &0.89 &0.74 &0.76 &0.89 &0.85 \\ \hline
      \end{tabular}
      \caption{Mean edge density in \mainproc{} output sets of size at least 10. Amazon has no output sets with $10$ vertices.}
      \label{tab:mean-density}
      \vspace{-0.75cm}
  \end{table*}
  
\paragraph{Formulation through triangle-rich sets.} Inspired by many results that exploit triangles to find
dense subgraphs, we define the notion of ``regularly triangle-rich" (RTR)  sets. These are subsets of vertices
of comparable degree that contain many triangles. These sets are significantly more restrictive than high edge density sets. The utility of this definition is that we can prove strong theoretical guarantees on algorithms. Our actual empirical goal is not recovering RTR sets, but we discover that it is a good guide for getting practical dense subgraph discovery algorithms.

\paragraph{Theoretical algorithm and analysis.} We can prove a strong "recovery" guarantee, using
the RTR framework. We design an algorithm, \mainproc, that outputs a disjoint family of RTR sets. A constant fraction of \emph{every} RTR set is contained in the family.
Under some stronger conditions on the RTR set, a constant fraction of the set is actually contained in a single output of \mainproc.
We borrow tools from theoretical work on triangle-dense
decompositions of graphs to design \mainproc~\cite{GuRoSe14}. 

\paragraph{Fast implementation of \mainproc.} We give a practical implementation of \mainproc, with some heuristics to improve the coverage. (The code is available at~\cite{rtr-code}.) As stated earlier, our aim is to cover a large portion of the graph with dense subgraphs.
We experiment on a large collection of real-world, publicly available
datasets. \mainproc{} runs extremely fast on large graphs and outputs 
a large collection of dense subgraphs on real-world graphs.

\paragraph{High coverage in practice.} For example, on a large {\tt Orkut} social
network with more than a hundred million edges, \mainproc{} covers a quarter of the graph with
subgraph of density more than $0.5$. We consistently observe this behavior across many datasets we experiment with. 
We compare \mainproc{} against a number of state of the art dense subgraph discovery algorithms and community detection procedures
(the Louvain algorithm~\cite{Louvain}, iterated {\tt Flowless}~\cite{Flowless}, iterated greedy~\cite{Charikar}, and the nucleus decomposition~\cite{SaSePi+15}). Nucleus
decompositions are the only algorithm specifically tailored to finding large families of dense subgraphs.
(We discuss more in \Sec{related}.) In \Fig{intro}, we compare the number of vertices covered
by each method, using sets of more than a given density. 
Across almost all datasets, \mainproc{} gets higher coverage of dense subgraphs compared to all these algorithms.

\paragraph{Finding many large dense subgraphs.} In \Fig{largest-scatter}, we show the size and density of the 20 largest sets output by \mainproc{} on various datasets. In all datasets,
we see that \mainproc{} can get numerous, large, dense sets. On the Orkut social network with 3M vertices, we get more than 20 sets of size more than 200, with edge density more than 0.3.
In \Tab{mean-density}, we give the mean density of output sets of size at least 10 vertices. We see the extremely high densities; in the case of the DBLP dataset, almost every output set is a clique (density of 0.1).

\paragraph{Qualitative examination of subsets.} We also demonstrate the semantic significance of the output family. \mainproc{} is able to output sets of similar vertices in labelled networks without any knowledge of the labels. For
a DBLP citation network dataset~\cite{aminer}, \mainproc{} outputs a family of sets of papers. We manually inspect these
sets and see that they are always on a single subtopic. These results reinforce the need to have dense subgraph 
discovery algorithms that output \emph{many} sets. 

\section{Related Work} \label{sec:related}

Finding the densest subgraph is a problem that has attracted a lot of attraction from both theoretical and applied researchers. It is of theoretical interest because most formulations are NP hard, and a lot of research has gone into approximation algorithms; indeed, there exist linear time approximation algorithms due to Asahiro~\cite{ASAHIRO200215}, Charikar~\cite{Charikar} and Tsourakakis~\cite{Tsou14} that have received a lot of attention. Other efficient algorithms in practice include more recent work in data mining ~\cite{Ts15, Flowless} leveraging cliques, quasi-cliques, as well as approaches based on cores and trusses ~\cite{AndC09, WZTT10, WC12, Huang_2017, XMFB23}.
A growing group of 
results use triangle information for algorithmic purposes~\cite{SaSePi+15,Ts15,BeGlLe16,TPM17, VGW18}. We point the readers to a number of surveys on dense subgraph discovery ~\cite{dense,LMFB23,Fang2020}.

A distinct but related problem is the one of community detection. While definitions of communities vary, it is widely accepted that dense subgraphs are closely related to real world communities in networks~\cite{LLDM08, CS12}.  One of the most celebrated algorithms in this field is the Louvain algorithm~\cite{Louvain}, which does a local maximization of the modularity metric due to Newman and Girvan~\cite{GiNe02}. The recent work of Miyauchi and Kawase~\cite{MK15}, and the later work of Miyauchi and Kakimura~\cite{MK18} use metrics related to modularity to find dense subgraphs in networks. Other popular methods include Infomap~\cite{infomap} and other approaches based on the map equation~\cite{Rosvall_2009}, the Leiden algorithm, ~\cite{Leiden}, label propagation~\cite{LabelProp}. The survey by Jin et al.~\cite{Jin_2021} takes a detailed look at different community detection methods from statistical modelling to the more recent deep learning based methods.

Most relevant to our work is the result of Gupta, Roughgarden, and Seshadhri~\cite{GuRoSe14}. They prove a decomposition theorem
for triangle-rich graphs, as measured by graph transitivity. Their main result shows that a triangle-dense graph
can be clustered into dense clusters. In a recent result, a spectral connection to this was found in the work of Basu, Bera and Seshadhri~\cite{basu2022spectral}, which involved generalizations of their proof technique using the normalized
adjacency matrix. Our main insight is that the algorithm of Gupta et al can adapted to extract our stronger notion of RTR sets. Moreover, we can prove explicit approximation guarantees on the output. We also note that while we present a scalable implementation, ~\cite{GuRoSe14} only give a theoretical algorithm and no implementation.

\section{The main problem} \label{sec:prob}

Consider an input simple graph $G = (V,E)$. Our high level objective is to output many
sets that are ``dense". We need to define our density objective. 

The typical notion of density
of a subset $S \subseteq V$ of vertices is the edge density $E(S)/{|S| \choose 2}$,
where $E(S)$ the number of edges contained in $S$. Note that this definition is interesting
only when $|S|$ is large, and hence algorithms finding useful dense subgraphs also need
to optimize for size. Another less common notion of density is $E(S)/|S|$, the average
degree inside $S$. While it is easy to optimize for this notion, the output often has
poor edge density and consists of extremely large sets. As we discuss in \Sec{experiments},
practical algorithms that optimize for these quantities give poor results (even in terms
of edge density).

Our starting point is that the best possible dense subgraph is an isolated clique. For example,
a (say) clique of size 15 formed by vertices of degree (say) 1000 is not an ideal dense subgraph.
We want our definitions to account for the degrees of the vertices involved in the dense set. While edge density captures the desired objective in practice, it is not mathematically ideal. First, the isolated edge is an optimum.
Moreover, suppose $S$ was (internally) a dense bipartite graph. It could have edge density
as high as $0.5$, but does not really conform with our usual notion of dense substructure. 
We would like high internal clustering coefficients as well. The intuition is repeatedly
observed, where using triangles leads to better density/community/clustering outcomes in practice ~\cite{SaSePi+15,Ts15,BeGlLe16,TPM17}.

Obviously, an isolated clique
of vertices is the best possible dense subgraph. We want our dense subgraphs to share
the characteristics of an isolated clique, as much as possible. These considerations
motivate our definition of \emph{\rtr{} sets}.

First, we give the standard definition of triangle density.

\begin{definition} \label{def:tri-dense} The triangle density of a set $S$ of vertices
is the number of triangles contained in $S$ divided by ${|S| \choose 3}$.
\end{definition}

The next definition is the central notion of this paper.

\begin{definition} \label{def:rtr} Let $\alpha \in [0,1]$. A set $S$ of vertices is called 
\emph{$\alpha$-\rtr} if (i) all vertices in $S$ have (total) degree in $[\alpha |S|,|S|/\alpha]$, and (ii) the triangle density of $S$ is at least $\alpha$.
\end{definition}

We note that triangle density implies edge density due to T\'uran-type theorems, so an $\alpha$-\rtr{} set also has
edge density $\alpha$.
Observe that the degree in the above definition refers to degree in the original graph, not just in the set $S$.
Thus, the total degrees are comparable to set size itself.
A $1$-\rtr{} set $S$ is necessarily an isolated clique, the perfect dense subgraph. One can think of $\alpha$
as a measure of how close a set $S$ is to being an isolated clique, where we account for both internal triangle density
and the degree regularity. Our focus is on the setting when $\alpha$ is a constant, like (say) $0.5$ or $0.8$.
For the sake of mathematical analysis, we will use the (standard) terminology $\Omega(1)$-\rtr{} to denote ``constant"-\rtr{} sets.

The primary utility of this definition is shown by our algorithmic results. The results
give guarantees on the output of \mainproc, our main procedure given in \Alg{mainproc}.
\mainproc{} outputs a family $\cT$ of disjoint sets. \mainproc{} can ``weakly discover" all $\Omega(1)$-\rtr{} sets in a graph. 

\begin{theorem} \label{thm:main} Consider an input graph $G = (V,E)$. For any constant $\alpha$,
there exists input parameters for the algorithm \mainproc{} 
with the following guarantees. (i) \mainproc{} outputs a disjoint family of sets $\cT$, such that each set
is $\Omega(1)$-\rtr. (ii) For \emph{any} $S$ that is $\alpha$-\rtr, a least a constant fraction
of $S$ is contained in $\cT$. (All constants have polynomial dependencies on $\alpha$.)

(The running time of \mainproc{} is, up to constant factors, the time taken to list all triangles
of $G$.)
\end{theorem}

This is a strong guarantee, since every $\Omega(1)$-\rtr{} set $S$ is approximately captured
by the output $\cT$. The output itself is a disjoint family of $\Omega(1)$-\rtr{} sets. Note that
we do not (and cannot) guarantee that $S$ is approximately captured by a single (or few) sets of $\cT$.
This is because \rtr{} sets can be overlapping, and for a disjoint family $\cT$ of \rtr{} sets,
some other $\Omega(1)$-rtr{} set is split among many sets of $\cT$.

Here are some examples showing these situations. Consider a graph formed by as follows. We first take $k$ disjoint $k$-cliques.
Then, we take a single vertex from each of these cliques, and form a $k$-clique $C$ inside it. There are now $(k+1)$ $k$-cliques,
and each of them is $1/2$-\rtr. (The $1/2$ is because the degrees are either $k-1$ or $2(k-1)$.) A perfectly reasonable
$\cT$ is to output the set of $k$ disjoint $k$-cliques. But then $C$ would be split among different sets in $\cT$. 
\Thm{main} asserts that, regardless of the graph structure, every $\Omega(1)$-\rtr{} set has a constant factor of its vertices in $\cT$,
which is itself a family of $\Omega(1)$-\rtr{} sets.

It is natural to ask if there are stronger assumptions under which an $\Omega(1)$-\rtr{} set is roughly
contained in one (or $O(1)$) sets in $\cT$. 

\begin{definition} \label{def:wellsep} Let $\alpha > \beta$. An $\alpha$-\rtr{} $S$
is called \emph{$\beta$-well separated} if: for all edges $(u,v)$ containing at least one
vertex in $S$, the edge contains at most $\beta |S|$ triangles where the third vertex
is outside $S$.
\end{definition}

So edges contained in a well-separated \rtr{} set have few triangles ``leaving" $S$. This 
is similar to a triangle cut constraint, used in previous work ~\cite{TPM17, BeGlLe16}.
For any well separated \rtr{} set $S$, we can prove a stronger condition: there exists an output 
set $\cT$ (of \mainproc) that contains a constant fraction of $S$.

\begin{theorem} \label{thm:wellsep} Consider any $\alpha$-\rtr{} set $S$ that is $\beta$-well separated,
for constants $\alpha, \beta$, where $\beta$ is sufficiently small (compared to $\alpha$).
Then, there exist input parameters for the algorithm \mainproc{} such that the output $\cT$
has a set that contains an $\Omega(1)$-fraction of $S$.
\end{theorem}

\paragraph{Connection to coverage.} As said earlier, in practice, we care for the coverage. \Thm{main} and \Thm{wellsep} do not explicitly give coverage guarantees, but there are implicitly implied. Suppose a large portion of the graph $G$ can be covered by disjoint $\Omega(1)$-\rtr{} sets. A constant fraction of \emph{each} of these sets is present in the output of \mainproc, so the output will have good coverage. (And each output set is \rtr, so the graph is covered with dense sets.) The "real" validation is done through our experiments in \Sec{experiments}.

\section{The Main Ideas Behind \mainproc{}}\label{sec:ideas}

We describe the main procedure, \mainproc. 
The inspiration for \mainproc{} is the theoretical decomposition results of Gupta, Roughgarden,
and Seshadhri~\cite{GuRoSe14}. (We stress that~\cite{GuRoSe14} is purely theoretical and has no
empirical results.) The main idea, which is itself present in numerous results, 
is to use the triangles to guide a local extraction of a dense ``cluster"~\cite{SaSePi+15,Ts15,BeGlLe16,TPM17}.

We use $d_v$ to denote
the degree of $v$ in the input graph $G$.

\begin{algorithm}[ht]
	\caption{\mainproc$(G,\eps, \beta)$: (We use $d_v$ to denote the degree of $v$ in the input graph $G$.)}
	\label{alg:mainproc}
	\begin{algorithmic}[1]
        \State Initialize an empty family $\cT$.
		\State Initialize subgraph $H$ to $G$.
		\While{$H$ is non-empty} 
            \While{there is edge $e = (u,v)$ in $< \eps(d_u + d_v)$ triangles in $H$}
                \State Delete $e$ from $H$. \label{step:clean}
            \EndWhile
            \State Delete all isolated vertices from $H$.
            \State Pick a vertex $v$ in $H$ with lowest $d_v$.
            \State Construct $N$, the neighborhood of $v$ in $H$.
            \State Let set $T = \{v\} \cup N$.
            \For{every vertex $u$ that is a neighbor of $N$}
                \State Let $t_u$ be the number of triangles (in $H$) from $u$ to $N$.
                \State If $t_u > \beta d^2_v$ , add $u$ to $T$. \label{step:twohop}
            \EndFor
            \State Delete $T$ from the subgraph $H$.
			\State Add set $T$ to output family $\cT$.
        \EndWhile
		\State Output $\cT$.
	\end{algorithmic}
\end{algorithm}

The procedure \mainproc{} takes two arguments $\eps, \beta$ that are thresholds used at
two separate steps. In the theorems, these (constant) parameters are chosen appropriate for the 
analysis. In practice, we set these to simple fixed values. Also, the practical implementation
differs somewhat from the theoretical description, but that is mostly for convenience.
We discuss these practical aspects in \Sec{prac}. In this section, our focus is on the theoretical aspects
and intuition for \mainproc.

\mainproc{} goes through an iterated extraction process. First, there is a ``cleaning" step (\Step{clean})
that removes edges in few triangles, yielding subgraph $H$. The main work is finding an \rtr{} set $T$
in $H$. This set is removed from $H$. The resulting $H$ is again cleaned, the next \rtr{} set is found,
and so on. We refer to the process of each output set $T$ being removed as an \emph{extraction}.

The cleaning operation of \Step{clean} removes edges that are present in too few triangles. 
Intuitively, these are edges that may ``distract" the subsequent extraction, so we delete them.
The parameter $\eps$ is used
to determine this threshold. Observe that $d_u + d_v$ is an upper bound on the number
of triangles that the edge $(u,v)$ participates in. \Step{clean} computes the fraction of triangles
that $(u,v)$ participates in with respect to this simple upper bound. That fraction is used as a threshold.
Note that the total number of triangles that $(u,v)$ participates in is at most $\min(d_u,d_v)$,
so this step automatically removes edges between vertices of disparate degrees.

The cleaning of \Step{clean}, while quite simple, is an immensely useful algorithmic (and practical)
idea. In a cleaned graph, the neighborhood of a vertex is automatically rich in triangles. So a simple
BFS based algorithm can extract out dense (and triangle-rich) sets. While this idea has appeared
in theory and some practical results ~\cite{SaPa11}, it has been underutilized as a technique
for dense subgraph discovery.

After cleaning, \mainproc{} takes the lowest degree vertex $v$ in the current subgraph $H$,
and start a ``seeded" extraction from $v$. At this point, all edges participate in sufficiently many triangles.
So the neighborhood $N$ of $v$ also contains enough triangles, and we set $T$ as $\{v\} \cup N$.
If we simply extract $T$ as a set, we run the risk of destroying too many triangle-rich sets.
The main insight is that we need to extract sets of radius $2$.

We give a simple justification. Suppose $G$ is a regular complete tripartite graph, so the vertices
are partitioned into three equal sized sets $V_1, V_2, V_3$, with all edges between sets. If $v \in V_1$,
then $\{v\} \cup N$ essentially consists of $V_2 \cup V_3$. Removing this set would destroy all the edges
(and triangles) incident to $V_1 \setminus \{v\}$. All vertices in $V_1$ would become singletons, which
is a bad output. More generally, dense subgraphs (like dense Erd\H{o}s-R\'{e}nyi graphs) often have radius
$2$. If $G$ was just a single uniformly dense random graph, extracting neighborhoods would not suffice.

So \mainproc{} looks at the neighbors of $N$, the two-hop neighborhood from $v$. The key insight
is to find vertices in the two-hop neighborhood that form sufficiently many triangles with the neighborhood $N$.
The argument $\beta$ is used as a threshold to determine which vertices are added to $T$. This 
parameter is critical to get two-hope neighborhoods that are triangle-rich. Vertices not assigned to any output set $T$ are considered singletons. 

\subsection{Challenges in the analysis} \label{sec:challenge}

There are two parts to the analysis. The first part shows that the output of \mainproc{} is a family
of \rtr{} sets. The second part shows that a given \rtr{} set $S$ is approximately captured,
either in weak sense of \Thm{main} or the stronger sense of \Thm{wellsep}. There is a tension
between these goals. To output \rtr{} sets, the cleaning of \Step{clean} needs to be aggressive
enough to delete edges that reduce triangle density. Also, \Step{twohop} adds vertices
to a set that is already triangle-rich, and hence we want it to be stringent.
But these are exactly the opposite of what is needed to preserve an existing triangle-rich set $S$.
The clean operation could delete edges from $S$, making vertices in $S$ disconnected in $H$.
Moreover, we would need \Step{twohop} to add in many vertices, to preserve as much of $S$.
The main point of the analysis is to prove that this tension can be resolved, with
the right arguments $\eps$ and $\beta$.

One of the insights of this analysis is that an accounting in terms of triangles
is more powerful than keeping track of edges and density. 
As we described earlier, \Step{clean} is a central operation to ensure that the output
is triangle-rich (and hence dense). Iterated with extraction steps, it might destroy
an existing \rtr{} $S$. As various extractions proceeds, they may remove a few vertices of $S$.
These removals affect the internal structure in $S$ by deleting triangles. Hence, subsequent
cleaning operations could end up deleting large portions of $S$, affecting the desired
guarantees of \Thm{main} and \Thm{wellsep}. 

We prove that the \rtr{} definition ensures that there is a ``core" of $S$ that is not
affected by deletions of vertices outside the core. The only way that \mainproc{}
can delete the core is by extracting vertices into the output sets $T$. Moreover,
using the well-separated condition of \Def{wellsep}, we can show any extraction
that removes even a single vertex of the core must remove a constant fraction of $S$.
These arguments crucially use \Step{twohop} to ensure that extractions that (say)
seeded in $S$ must cover a significant fraction of $S$.

\section{Analysis}\label{sec:analysis}

The first step in the analysis is to show that every set output by \mainproc is $\Omega(1)$-\rtr{}.
For convenience, we think of the input arguments $\eps, \beta$ to constants, so we use $\Omega(1)$
notation to suppress dependencies on these values. We also consider the $\alpha$ parameter (in \Def{rtr})
to be a constant. All constants will have at most polynomial dependencies on each other. We do not
try to optimize these dependencies, so our mathematical analysis is focused on the asymptotics.
In practice, we set these parameters to be $0.1, 0.3, 0.5$, etc. So assuming they are constant is consistent
with our experiments.
(The true ``test" of \mainproc{} is the strong experimental results.)

\begin{lemma} \label{lem:rtr} Every $T$ output by \mainproc$(\eps,\beta)$ is $\Omega(\poly(\eps,\beta))$-\rtr.
Hence, if $\eps, \beta = \Omega(1)$, then every $T$ output is $\Omega(1)$-\rtr.
\end{lemma}

\begin{proof} Consider some output set $T$. Suppose it is removed from subgraph $H$, starting
from a seed vertex $v$.  In all calculations that follow, we are considering the properties
of the subgraph $H$. A crucial property is that every edge $(x,y)$ in $H$
participates in at least $\eps(d_x + d_y)$ triangles. Moreover, all degrees are at least $d_v$.

Since $v$ has non-zero degree in $H$, there is some edge $(u,v)$ in $H$.
The edge $(u,v)$ participates in at least $\eps(d_u + d_v)$ triangles in $H$. Observe that 
the number of triangles is at most $\min(d_u, d_v) = d_v$. Thus, $\eps(d_u + d_v) \leq d_v$,
implying $d_u \leq d_v/\eps$. In general, for any edge $(x,y)$ in $H$, $d_x = \Theta(d_y)$. For any arbitrary vertex in the two-hop neighborhood, the degree is at most $d_v/\eps^2$.

Since $(u,v)$ participates in more than $\eps d_v$ triangles in $H$, $v$ must have at least $\eps d_v$
neighbors in $H$. Moreover, every edge $(u',v)$ participates in at least $\eps d_v$ triangles.
Thus, $v$ participates in at least $\eps^2 d^2_v/2$ triangles ($\eps d_v$ triangles along every edge $(u',v)$, 
and each triangle is counted at most twice). Every triangle incident to $v$ corresponds to an edge
in the neighborhood $N$ of $v$. So $N$ contains $\eps^2 d^2_v/2$ edges. Every edge
participates in at least $\eps d_v$ triangles. Hence, there are at least $\eps^3 d^3_v/2$ triangles
with at least two vertices in $N$.

Hence, $\sum_u t_u \geq \eps^3 d^3_v/2$. For every $u$, let $b_u$ be the number of neighbors of $u$
in $N$. Observe that $t_u \leq b^2_u$, since every triangle that $u$ forms with $N$
is made by a pair of neighbors in $N$. Thus, $\sum_u \sqrt{t_u} \leq \sum_u b_u$.
The term $\sum_u b_u$ counts the number of edges incident to $N$. The size of $N$ is at most $d_v$,
and each vertex in $N$ has degree at most $d_v/\eps$ (as shown earlier). So $\sum_u \sqrt{t_u} \leq \sum_u b_u \leq d^2_v/\eps$.

Consider the sum $\sum_{u: t_u < \eps^8 d^2_v/9} t_u$. We can bound this as follows:
\begin{eqnarray*} \label{eq:sqrt-bd}
\sum_{u: t_u < \eps^8 d^2_v/9} t_u = \sum_{u: t_u < \eps^8 d^2_v/9} \sqrt{t_u} \sqrt{t_u} & \leq & (\eps^4/3) d_v \sum_u \sqrt{t_u} \\
& \leq & (\eps^4/3) d_v \cdot d^2_v/\eps = \eps^3 d^3_v/3 
\end{eqnarray*}
Thus, for $\beta = \eps^8/9$,  we get that $\sum_{u: t_u \geq \beta d^2_v} t_u \geq \eps^3 d^3_v/2 - \eps^3 d^3_v/3 = \Omega(d^3_v)$. 
Hence, \Step{twohop} picks up $\Omega(d^3_v)$ triangles inside the set $S$. To bound the triangle density
of $S$, we need to upper bound to size $S$. The size of $N$ is at most $d_v$. Each vertex in $N$ has degree $\Theta(d_v)$,
so there are $\Theta(d^3_v)$ triangles incident to $N$. Hence, $\sum_u t_u = O(d^3_v)$. Each $u$ added to $S$
has $t_u = \Omega(d^2_v)$, so at most $O(d_v)$ such vertices are added to $S$ by \Step{twohop}. We conclude
that $S$ has $O(d_v)$ vertices, so the triangle density of $S$ is $\Omega(1)$.

Finally, we bound the degrees in $S$. All vertices are within distance $2$ of the seed vertex $v$. 
As proven above, for any pair of neighboring vertices, their degrees are within constant factors of each other.
Hence, all degrees in $S$ are $\Theta(d_v)$ and $S$ has size $O(d_v)$. So $S$ is $\Theta(1)$-\rtr.
\end{proof}

We have proved the first part of \Thm{main}. For the second part, we need to show that 
a constant fraction of any $\Omega(1)$-\rtr{} set is contained in $\cT$. 

\begin{lemma} \label{lem:single} Consider any $\alpha$-\rtr{} set $S$. Then, there exists 
a setting of $\eps, \beta$ as $\poly(\alpha)$, such that an $\Omega(\alpha)$
fraction of vertices of $S$ is present in the family $\cT$. (Where $\cT$ is the
output of \mainproc$(\eps, \beta)$.)
\end{lemma}

\begin{proof} Edges from $G$ are removed in two steps. First, in $H$, \Step{clean} removes edges.
Second, when a vertex is put into an output cluster, all incident edges are removed. 
Suppose $S$ is $\alpha$-\rtr{}.
Assume, for contradiction's sake, that at most $\alpha |S|/8$ vertices of $S$
are put in output clusters. Let $s$ denote the size of $S$; note that all degrees
of vertices in $S$ lie in $[\alpha s,s/\alpha]$. We will set $\eps$ to be (anything) $\leq \alpha^2/8$.

Instead of accounting for vertices, we will look at the triangles contained in $S$ (call these $S$-triangles). A triangle gets removed
if any of its vertices is removed. Let us count the number of $S$-triangles removed
by edge deletions in \Step{clean}. When such an edge $(u,v)$ is removed, it participates in at most $\eps(d_u + d_v) \leq \alpha^2/8 \cdot s/\alpha =$
$\alpha d/4$ triangles.
There are at most $s^2/2$ edges in $S$, so all these operations remove at most $s^2/2 \times \alpha s/4 = \alpha ds^2/8$ triangles.

Now, we bound the number of $S$-triangles removed when vertices are put into output clusters. By assumption, there are 
at most $\alpha s/8$ such vertices. Each vertex participates in at most $s^2/2$ $S$-triangles. So there are
at most $\alpha s^3/16$ triangles. In total, at most $(1/8 +1/16) \alpha s^3$ triangles are removed.

On the other hand, since $S$ is $\alpha$-triangle rich, there are at least $\alpha {s \choose 3} \geq \alpha s^3/6$ triangles.
This is more than the number of triangles removed by \mainproc. Observe that all triangles are removed by the end of \mainproc,
so this is a contradiction. Hence, at least an $\Omega(\alpha)$ fraction of vertices of $S$ must be put in output clusters.
\end{proof}

The previous two lemmas prove \Thm{main}. We show that the output of \mainproc{} is $\Omega(1)$-\rtr,
and that a constant fraction of every $\Omega(1)$-\rtr{} set $S$ is contained in the output.

\subsection{The well separated case} \label{sec:wellsep}

As argued earlier, we need some stronger conditions that prove that a \emph{single} output
set contains a constant fraction of $S$. This stronger condition is that of well-separated,
given in \Def{wellsep}, and the corresponding theorem is \Thm{wellsep}.

This proof of \Thm{wellsep} is more complicated, since we need to argue that there is an extraction step
that takes out a constant fraction of $S$. The proof of \Lem{single} shows that 
there are extractions that remove a constant fraction of $S$. For \Thm{wellsep},
we need to argue that there is a \emph{single} extraction that does the above.
This requires a careful analysis of how vertices in $S$ get deleted from
the current subgraph $H$. The parameter $\beta$ plays an important role in the proof.
The threshold in \Step{twohop} is used to prevent many vertices of $S$ ending up
in other extractions. Roughly speaking, only an extraction seeded inside $S$ can remove
a significant part of $S$.

We encapsulate the analysis as the following lemma. From the lemma, \Thm{wellsep} follows
directly.

\begin{lemma} \label{lem:wellsep} Consider any $\alpha$-\rtr{} set $S$ that is $\delta$-well separated,
where $\delta$ is sufficiently smaller than $\alpha$. Then, there exists 
a setting of $\eps, \beta$ as $\poly(\alpha)$, such that an $\Omega(\poly(\alpha))$
fraction of vertices of $S$ is present in the family $\cT$. (All constants are polynomially
related to each other.)
\end{lemma}

\begin{proof} (of \Thm{wellsep}) By \Lem{single}, a non-zero number of vertices of $S$ are
present in output clusters. Consider the first vertex of the any output set $T$
that intersects $S$. (There must exist some such set.) 
Following the notation in \mainproc, $v$ denotes the seed vertex of this output set. Then there are three possibilities: either $v\in S$, or $N(v)\cap S\neq \varnothing,$ or a vertex from $S$ is added to $T$ in \Step{twohop}.
We handle each case separately. We assume that $\delta < \alpha^2/4$, for a sufficiently large 
constant $c$. We set $\eps = \alpha^2$ and $\beta$. 

{\em Case 1, $v \in S$:} At this stage, $v$ has non-zero degree in $H$.
So there is an edge $(u,v)$ containing at least $\eps d_v$ triangles in $H$.
Since $S$ is $\delta$-well separated, there are at most $\delta d$ triangles (containing $u,v$)
with a third vertex outside $S$. 
Thus, in $H$, there are at least $(\eps - \delta)d \geq \eps d/2$ triangles (containing $u,v$)
whose third vertex is inside $S$ for all $k\geq 3$. (Recall, we set $\eps = \alpha^2$ and $\delta < \alpha^2/4$.)
All of these ``third vertices" are obviously neighbors
of $v$. Hence, the neighborhood $N$ of $v$ will contain at least $\Omega(\eps d_v)$ vertices of $S$.
The output cluster contains all of these vertices.
Since $S$ is $\alpha$-\rtr, $d_v \geq \alpha |S|$. Hence, the neighborhood $N$ contains $\Omega(\poly(\alpha)|S|)$ vertices
of $S$.

{\em Case 2, a neighbor of $v$ is in $S$:} Suppose this neighbor is $u$. The edge $(v,u)$
is in $H$. By the same argument as the previous case, this edge has at least $\Omega(\eps d_v)$
triangles in $H$ with the third vertex in $S$. All of these vertices are neighbors of $v$,
and will be part of the output cluster.

{\em Case 3, \Step{twohop} removes a vertex from $S$:} We can assume that neither Case 1 or 2
holds. Hence, $v \cup N$ lies outside $S$. Suppose $u \in S$ is added to the output cluster.
The vertex $u$ must participate in at least $\beta d^2_v$ triangles whose other vertices are in $N$.
Obviously, each of these triangles participates in an edge $(u,w)$ (for $w \in N$). Since
$S$ is well-separated, each such edge $(u,w)$ can be in at most $\delta |S|$ triangles
whose third vertex is outside $S$. 

Since the size of $N$ is at most $d_v$, there are at most $\delta d_v|S|$ triangles
that $u$ can form with two vertices in $N$. Note that $d_u \leq d_v/\eps^2$ (since $u$ is distance
two away from $v$ after cleaning), and $|S| \leq d_u/\alpha$ ($S$ is $\alpha$-\rtr{} and $u \in S$).
Thus, $\delta d_v|S| \leq (\delta/(\alpha\eps^2)) d^2_v$. Since we set $\beta = \alpha^c$ for sufficiently large
constant $c$, we get that $\delta d_v|S| < \beta d^2_v$. 

So $u$ does not participate in enough triangles to be selected in \Step{twohop}. This case cannot happen.
In other words, \Step{twohop} cannot remove a vertex from $S$.
We conclude that Case 3 cannot happen.
\end{proof}

\subsection{Running time} \label{sec:runtime}

We provide an upper bound on the running time of the algorithm as presented in \Alg{mainproc}. The most expensive step is a triangle
enumeration step to initialize data structures. The best known practical method for triangle enumeration is the classic 
algorithm of Chiba-Nishizeki~\cite{ChNi85}, and variants by Schank-Wagener~\cite{ScWa05}. The running time of the Chiba-Nishizeki
algorithm is $O(m\alpha)$, where $\alpha$ is the graph degeneracy (Sec 2.4 of~\cite{Se23}). In practice, graph degeneracies
are small and this algorithm is extremely efficient~\cite{Se23}.

\begin{theorem} \label{thm:runtime}
    The running time of \mainproc{} is $O(t + (m+n)\log n)$, where $t$ is the running time of triangle enumeration. 
\end{theorem}

\begin{proof}
    Assume that $G$ is given as an adjacency list. Each adjacency list is stored as a dictionary data structure, for convenience. At every iteration, we maintain the following data structures about our current subgraph $H$. 
    \begin{asparaitem}
        \item A list of triangles $T_H$ that are contained in $H$, indexed by edge. So for each edge $e$, we have a dictionary data
structure containing all the triangles in $H$ that contain $e$.
        \item A list of bad edges $B$ that do not satisfy the criterion in \Step{clean}. 
        \item A min priority queue of all vertices keyed by degree. 
    \end{asparaitem}
    The time to initialize these data structures is $O(t+m)$. We perform a triangle enumeration in $G$ to 
get the number of triangles on each edge ($T_G$). The list $B$ can be initialized by looping over all edges.

    Consider the removal of edge $(u,v)$ in \Step{clean}. For every triangle of the form $(u,v,w)$ that it participates in, we must remove this triangle from the corresponding list in $T_H$
 of every edge $e'$ in it ($e'$ is of the form $(u,w)$ or $(v,w)$). On performing the removal, we check if $e'$ is now a bad edge, to add to $B$. If the removal causes $u$ or $v$ to not have any neighbors in $H$, then we remove the vertex from the priority queue. Deletion/addition in a priority queue takes time $O(\log n)$. Besides that, we have a constant number of operations in the number of triangles on the bad edge $e$. All removals due to \Step{clean} and isolated vertex deletions take time $O(t + n\log n)$.

    The next step is to pick the lowest degree vertex $v$ from our priority queue. Constructing $N$, the neighborhood of $v$, takes $O(d_v)$ time by traversing the adjacency list. 
Consider the set of edges contained in $\{v\}\cup N$. For each edge, we look up the triangles incident on it via $T_H$ and consider these vertices. The number of distinct triangle-forming vertices is at most the number of triangles in $H$ with two vertices in $\{v\} \cup N$. We can then check the \Step{twohop} condition for each such vertex, and add the required vertices to our output set. 
We now remove the output set $T$ from the subgraph $H$. This requires changes to the adjacency list, the list of bad edges, and our priority queue. 
These are all implemented by deleting from $H$ all the remaining edges incident to every vertex in $T$.
The number of vertices removed from the priority queue is $|T|$, 
and the number of changes to the data structures is at most $O(\sum_{u\in T} d_u)$, the number of edges incident to vertices in $T$. 
Each deletion operation takes time at most $O(\log n)$. Combined over all output sets, the running time of all of these operations
is at most $O(\sum_T \sum_{u \in T} d_u\log n + n\log n) = O(m+n)\log n$. The total running time, including
the data structure initialization, is $O(t + (m+n)\log n)$.
\end{proof}

\section{Running \mainproc{} on Real Networks} \label{sec:prac}

In this section, we discuss aspects of the implementation of \mainproc. These are a collection of heuristics and parameter
settings to improve empirical performance.

\paragraph{The growing heuristic:} Recall that the output of \mainproc{} is a collection $\cT$ of disjoint \rtr{} sets.
There are vertices that are not captured in any of these sets, and lead to a loss of coverage. We take inspiration
from the proof of \Thm{wellsep} to design a simple growing heuristic. In the proof of \Thm{wellsep}, observe
that Case 3 can never happen. For a given well separated set $S$, there is a set $T$ that contains a constant fraction
of $S$, such that the seed vertex for $T$ is either in $S$ or a neighbor (it cannot be at distance $2$). Intuitively,
we might expect that most of $S$ is present within distance $1$ of $T$. Towards that, we perform the following
growing heuristic that takes a parameter $k$. for every $v$ unassigned to a set in $\cT$,
add $v$ to a $T$ such that $v$ has at least $k$ neighbors in $T$. (If there are multiple, choose the set with the most neighbors.)
This heuristic increases the coverage of the output, at the cost of reducing the edge density. In practice,
we see a negligible reduction in density and increased coverage. For all our experimental results,
we use this heuristic on top of \mainproc, setting $k$ to be $10$.

\paragraph{The choice of $\eps$:} The parameter $\eps$ used in \Step{clean} plays a central role
in the algorithm. In practice, we set $\eps = 0.1$ for all datasets. We do not see significant differences in varying $\eps$ between $0.05$ and $0.3$.

\paragraph{Doing away with $\beta$ by a greedy sweep:} Recall that $\beta$ is used to set the threshold in \Step{twohop}. 
In practice, we follow a greedy heuristic to pick the two-hop vertices added to $T$. This heuristic is inspired
by the proof of \Lem{rtr}. In the analysis for the two-hop vertices, we use partial sums based on thresholding
by $t_u$, the number of triangles $u$ forms with $N$. In our implementation of \mainproc, we first sort the vertices outside $N$
in decreasing order of the number of triangles. We add the vertices to $N$ in order, and keep track of the current edge density.
Finally, we pick the prefix of (the sorted) vertices whose addition maximizes the edge density.

\paragraph{Improving the space complexity with triangle recomputation:} This is an implementation trick
that improves on the proof of \Thm{runtime}. As stated, the proof requires a storage of all the triangles.
We do away with this data structure by only storing the triangle counts on edges. Every time we need to list
of triangles on an edge, we simple recompute these triangles. On a careful look at the implementation,
we can observe that this recomputation is only required a constant number of times per edge.
So we only lose a constant factor on the running time, but improve storage to $O(m)$. This leads
to significant improvements in practical running time.

\section{Experimental results} \label{sec:experiments}
\begin{table}
    \centering
    \begin{tabular}{|c|c|c|c|c|c||c|H|}
    \hline
         Network & $|V|$ & $|E|$ & $\#T$ & $\overline{d}$ & ED & \#Sets & \#Sing. \\
         \hline
         EUEmail & 1k & 16k & 105k & 16 & 3.2e${-2}$& 161 & 384\\
         Hamsterster & 2.4k & 16.6k & 53k & 7  & 5.6e${-3}$ & 271 & 941\\
         CaAstroPh & 18.7k & 198k & 1.4M & 11  & 1.1e${-3}$ & 2k & 7k\\
         Epinions & 76k & 406k & 1.6M& 5 & 1.4e${-4}$ & 1.2k & 71k\\
         Amazon & 335k & 926k & 667k & 3 & 1.7e${-5}$ & 38k & 192k\\
         DBLP & 317k & 1M & 2.2M & 3 & 2.1e${-5}$ & 57k & 119k\\
         Youtube& 1.1M & 3M & 3M & 3 &4.6e${-6}$ & 15k & 1M\\
         Skitter& 1.7M & 11M & 29M & 7 & 7.7e${-6}$ & 40k & 1.5M\\
         Wikipedia & 1.8M & 25M & 51M & 14 & 1.6e${-6}$ & 31k & 1.6M\\
         Livejournal & 4M & 35M & 178M & 9 & 4.3e${-6}$ & 188k & 3M \\
         Orkut & 3M & 117M & 628M & 38 & 2.5e${-5}$ & 57k & 2.3M \\
         \hline
         
         \end{tabular}

    \caption{A summary of all the datasets used in our evaluation. We list, respectively, the number of vertices, edges, triangles, average degree, and edge density of the network. In the last column, we list the number of sets output by \mainproc.}
    \label{tab:summary}
    \vspace{-1cm}
\end{table}
\begin{table*}[]

\begin{tabular}{|l||rrrrrl||rrrrrl|}
\hline

\multicolumn{13}{|c|}{Coverage: Percentage of vertices in sets of size $\geq 5$ with density above a threshold}\\
        \hline
        Threshold & \multicolumn{6}{c||}{0.5}                    & \multicolumn{6}{c|}{0.8}                    \\ \hline 
Network & \multicolumn{1}{l|}{RTREx}   & \multicolumn{1}{l|}{Louvain} & \multicolumn{1}{l|}{$k$-Truss}   & \multicolumn{1}{l|}{$k$-Core} & \multicolumn{1}{l|}{Greedy} & Flowless & \multicolumn{1}{l|}{RTREx}   & \multicolumn{1}{l|}{Louvain} & \multicolumn{1}{l|}{$k$-Truss}   & \multicolumn{1}{l|}{$k$-Core} & \multicolumn{1}{l|}{Greedy} & Flowless          \\ \hline \hline
EUEmail          & \multicolumn{1}{r|}{\textbf{33.40\%}} & \multicolumn{1}{r|}{2.99\%}           & \multicolumn{1}{r|}{8.46\%}           & \multicolumn{1}{r|}{9.05\%}        & \multicolumn{1}{r|}{2.99\%}          & \multicolumn{1}{r||}{0.00\%} & \multicolumn{1}{r|}{\textbf{23.20\%}} & \multicolumn{1}{r|}{0.00\%}           & \multicolumn{1}{r|}{4.48\%}           & \multicolumn{1}{r|}{0.00\%}        & \multicolumn{1}{r|}{0.00\%}          & \multicolumn{1}{r|}{0.00\%} \\ \hline
Hamsterster      & \multicolumn{1}{r|}{\textbf{45.20\%}} & \multicolumn{1}{r|}{32.80\%}          & \multicolumn{1}{r|}{24.90\%}          & \multicolumn{1}{r|}{6.63\%}       & \multicolumn{1}{r|}{1.36\%}          & \multicolumn{1}{r||}{4.49\%} & \multicolumn{1}{r|}{\textbf{43.50\%}} & \multicolumn{1}{r|}{14.00\%}          & \multicolumn{1}{r|}{21.60\%}          & \multicolumn{1}{r|}{6.43\%}       & \multicolumn{1}{r|}{0.87\%}          & \multicolumn{1}{r|}{1.90\%} \\ \hline
CaAstroPh        & \multicolumn{1}{r|}{\textbf{47.16\%}} & \multicolumn{1}{r|}{27.78\%}          & \multicolumn{1}{r|}{19.23\%}          & \multicolumn{1}{r|}{1.75\%}        & \multicolumn{1}{r|}{0.64\%}          & \multicolumn{1}{r||}{3.59\%} & \multicolumn{1}{r|}{\textbf{46.82\%}} & \multicolumn{1}{r|}{8.35\%}          & \multicolumn{1}{r|}{14.87\%}           & \multicolumn{1}{r|}{1.04\%}        & \multicolumn{1}{r|}{0.60\%}          & \multicolumn{1}{r|}{2.09\%} \\ \hline
Epinions         & \multicolumn{1}{r|}{2.87\%}           & \multicolumn{1}{r|}{1.74\%} & \multicolumn{1}{r|}{\textbf{3.03\%}}           & \multicolumn{1}{r|}{0.00\%}        & \multicolumn{1}{r|}{0.17\%}          & \multicolumn{1}{r||}{0.79\%} & \multicolumn{1}{r|}{\textbf{2.57\%}}           & \multicolumn{1}{r|}{0.11\%} & \multicolumn{1}{r|}{1.55\%}           & \multicolumn{1}{r|}{0.00\%}        & \multicolumn{1}{r|}{0.00\%}          & \multicolumn{1}{r|}{0.21\%} \\ \hline
Amazon           & \multicolumn{1}{r|}{42.70\%}          & \multicolumn{1}{r|}{43.80\%}          & \multicolumn{1}{r|}{\textbf{45.50\%}} & \multicolumn{1}{r|}{4.76\%}        & \multicolumn{1}{r|}{0.00\%}          & DNF &\multicolumn{1}{r|}{\textbf{40.70\%}} & \multicolumn{1}{r|}{14.52\%}          & \multicolumn{1}{r|}{28.04\%}          & \multicolumn{1}{r|}{2.50\%}        & \multicolumn{1}{r|}{0.00\%}          & DNF \\ \hline
DBLP             & \multicolumn{1}{r|}{10.00\%} & \multicolumn{1}{r|}{27.30\%}          & \multicolumn{1}{r|}{\textbf{34.40\%}}          & \multicolumn{1}{r|}{1.74\%}        & \multicolumn{1}{r|}{0.21\%}          & DNF &\multicolumn{1}{r|}{9.98\%} & \multicolumn{1}{r|}{7.03\%}          & \multicolumn{1}{r|}{\textbf{25.20\%}}          & \multicolumn{1}{r|}{1.61\%}        & \multicolumn{1}{r|}{0.20\%}          & DNF \\ \hline
Youtube          & \multicolumn{1}{r|}{0.48\%}           & \multicolumn{1}{r|}{1.37\%} & \multicolumn{1}{r|}{\textbf{2.37\%}}           & \multicolumn{1}{r|}{0.02\%}        & \multicolumn{1}{r|}{0.00\%}          & DNF &\multicolumn{1}{r|}{0.44\%}           & \multicolumn{1}{r|}{0.04\%} & \multicolumn{1}{r|}{\textbf{0.82\%}}           & \multicolumn{1}{r|}{0.02\%}        & \multicolumn{1}{r|}{0.00\%}          & DNF \\ \hline
Skitter          & \multicolumn{1}{r|}{2.57\%}           & \multicolumn{1}{r|}{1.79\%}           & \multicolumn{1}{r|}{\textbf{3.11\%}}  & \multicolumn{1}{r|}{0.00\%}        & \multicolumn{1}{r|}{0.00\%}          & DNF &\multicolumn{1}{r|}{\textbf{1.56\%}}           & \multicolumn{1}{r|}{0.07\%}           & \multicolumn{1}{r|}{0.67\%}  & \multicolumn{1}{r|}{0.00\%}        & \multicolumn{1}{r|}{0.00\%}          & DNF \\ \hline
Wikipedia        & \multicolumn{1}{r|}{\textbf{5.38\%}}  & \multicolumn{1}{r|}{0.28\%}           & \multicolumn{1}{r|}{1.00\%}           & \multicolumn{1}{r|}{0.00\%}        & \multicolumn{1}{r|}{0.00\%}          & DNF &\multicolumn{1}{r|}{\textbf{4.21\%}}  & \multicolumn{1}{r|}{0.11\%}           & \multicolumn{1}{r|}{0.59\%}           & \multicolumn{1}{r|}{0.00\%}        & \multicolumn{1}{r|}{0.00\%}          & DNF \\ \hline
Livejournal      & \multicolumn{1}{r|}{\textbf{17.80\%}} & \multicolumn{1}{r|}{4.86\%}          & \multicolumn{1}{r|}{12.70\%}          & \multicolumn{1}{r|}{0.38\%}        & \multicolumn{1}{r|}{0.02\%}          & DNF &\multicolumn{1}{r|}{\textbf{16.04\%}} & \multicolumn{1}{r|}{0.87\%}           & \multicolumn{1}{r|}{8.21\%}          & \multicolumn{1}{r|}{0.35\%}        & \multicolumn{1}{r|}{0.02\%}          & DNF \\ \hline
Orkut            & \multicolumn{1}{r|}{\textbf{24.20\%}} & \multicolumn{1}{r|}{0.34\%}           & \multicolumn{1}{r|}{8.98\%}          & \multicolumn{1}{r|}{0.00\%}        & \multicolumn{1}{r|}{0.00\%}          & DNF&\multicolumn{1}{r|}{\textbf{16.10\%}}          & \multicolumn{1}{r|}{0.04\%}           & \multicolumn{1}{r|}{5.68\%} & \multicolumn{1}{r|}{0.00\%}        & \multicolumn{1}{r|}{0.00\%}          & DNF \\ \hline
\end{tabular}
    \caption{For each dataset, we look at the sets/clusters output by each method, and their coverage at different densities.
    Coverage at 0.5 (respectively 0.8) is the percentage of the vertex set of the graph that is contained in output sets of density more than 0.5 (respectively 0.8). We restrict our attention to sets with more than 5 vertices. Only Louvain and $k$-truss do well. In cases where they are ahead, the difference is typically small (except DBLP). At the higher threshold, \mainproc{} usually has a significant advantage. The numbers in bold are the highest coverage figures at each threshold. DNF denotes the method did not finish.}
    \vspace{-0.75cm}
 \label{tab:coverage}
 \end{table*}

We perform a detailed empirical evaluation of \mainproc, by running it on a large variety of publicly available datasets.
We also compare with a number of state of the art procedures. 
We simply use the term ``density" to denote edge density. Our aim is to get a large collection of dense sets.
\paragraph{Implementation details:} Our code for \mainproc{} is at \url{https://github.com/amazon-science/amazon-RTRExtractor/tree/main}.  All code is written in C++17, and run on a PC with an Intel i7-10750H. When running \mainproc{}, the value of $\eps$ is set to $0.1$ and the parameter of growing heuristic is set to $10$. For our experiments, the networks have been taken from the SNAP database~\cite{snapnets}, with the exception of the labelled DBLP citation network (V2) taken from \url{aminer.org}~\cite{aminer}. We summarize the datasets, important statistics of the networks, and some primary features of our results in \Tab{summary}. Our implementation of \mainproc{} is fast. Even in the largest network in our dataset, \mainproc{} runs in minutes on a laptop.
\paragraph{Comparisons with other algorithms:} There is no explicit algorithm that tries to maximize coverage and edge density. Nevertheless, 
as discussed in \Sec{related}, there are numerous procedures that find dense subgraphs. Each procedure outputs a family (or hierarchy) of sets. We do not experiment with methods that cannot scale to graphs with hundreds of 
millions of edges. We compare \mainproc{} with four state of the art procedures. 

 \begin{asparaitem}
     \item The $k$-core decomposition: a linear $(O(m))$ time decomposition algorithm ~\cite{MB83} that creates a hierarchical clustering of vertices by iteratively removing vertices of degree less than $k$.
     \item The $(2,3)$-nucleus decomposition~\cite{SaSePi+15}: This is an important algorithm that finds a large set of dense subgraphs and is arguably the closest
         to our objective. We choose the (2,3)-nucleus, which is exactly the $k$-truss decomposition~\cite{truss}. The decomposition is a hierarchy, so it does not give an explicit disjoint collection. So, to analyze density $\gamma$, we take the highest
         sets in the hierarchy with density more than $\gamma$. We note that sibling nuclei can share vertices, but are edge disjoint. We refer to this as $k$-truss.
     \item The Louvain algorithm~\cite{Louvain}: While community detection methods do not try to optimize for edge density, it is instructive
         to compare with arguably the most famous community detection method. Despite the plethora of such methods, the Louvain algorithm
         is still one of the most scalable procedures that usually does well.
     \item Iterated Charikar's greedy peeling algorithm~\cite{Charikar}: This greedy heuristic is a method to get a set with large average degree,
         and has been used as the basis of many dense subgraph discovery algorithms ~\cite{Tsou14, Ts15, Flowless, CQT22}. The output is only a single set, so we apply
         it repeatedly to get a collection of dense subgraphs. We denote this algorithm as Greedy.
     \item Iterated Flowless~\cite{Flowless}: The Flowless algorithm optimizes the average degree of a subgraph as well.
         We iterate repeatedly to get a collection of sets and refer to the procedure as Flowless.
 \end{asparaitem}

If an algorithm does not finish in twelve hours, we terminate it and note DNF in our results. We note that Flowless does not scale well, and times out for large
instances.

\paragraph{Measuring coverage:} For any family of disjoint sets, we will often measure the \emph{coverage at a density}. This means, for a given density $\gamma$, we look
at the total number of vertices in sets of size at least 5 and density at least $\gamma$. In \Tab{coverage} we take a detailed look at coverage for all our datasets. In \Fig{intro}, for all the methods, we remove sets of extremely small size ($\leq 5$ vertices). This thresholding
helps remove trivial outputs and disconnected components.
The other metric of importance is the density of the top (say) 20 sets: we look at this in \Fig{largest-scatter}. 
We now list out the major experimental findings.

\begin{table*}
    \begin{minipage}{.6\linewidth}
        {\begin{tabular}{|c|c|c|c|c|c|c|c|c|}
    \hline
         \multirow{2}{*}{Network} & \multicolumn{2}{c|}{\mainproc{}}& \multicolumn{2}{c|}{Louvain} & \multicolumn{2}{c|}{Greedy} &
         \multicolumn{2}{c|}{Flowless}\\ \cline{2-9} & Size & ED  & Size & ED & Size & ED & Size & ED\\ \hline
         EUEmail & 16 & 0.93 & 143 & 0.14 & 227 & 0.24 & 244 & 0.22\\
         Hamsterster & 17 & 1.00 & 143 & 0.10 & 435 & 2.5e-3 & 254 & 0.04\\
         CaAstroPh & 36 & 0.98 & 1.1k & 0.02 & 2.7k & 3.5e-3 & 1.2k & 0.01\\
         Epinions & 40 & 0.69 & 11k & 1.4e-3 & 24k & 5.8e-5 & 9.1k & 1.5e-4\\
         Amazon & 7 & 1.00 & 318 & 0.02 & 60k & 2.3e-5 & DNF&DNF \\
         DBLP & 59 & 1.00 & 7k & 1.1e-3 & 64k & 1.9e-5 & DNF&DNF \\
         Youtube & 22 & 0.79 & 95k & 9.5e-5 & 377k & 3.9e-6 &DNF &DNF\\
         Skitter & 102 & 0.23 & 114k & 1.3e-4 & 379k & 3.5e-6 & DNF& DNF\\
         Wikipedia & 341 & 0.20 & 181k & 1.1e-4 & 253k & 5.6e-6 &DNF &DNF \\
         Livejournal & 342 & 0.40 & 508k & 2.7e-5 & 655k & 2.1e-6 &DNF &DNF \\
         Orkut & 383 & 0.32 & 543k & 7.2e-5 & 454k & 1.0e-4 &DNF &DNF \\
         \hline
        \end{tabular}
        \caption*{Largest output set and their densities}}

    \end{minipage}
         \hfill
\begin{minipage}{.35\linewidth}
{    \begin{tabular}{|c|c|c|c|}
\hline
Network & RTREx & Louvain & $k$-Truss \\ \hline
EUEmail          & 0.08     & 0.03       & 1.53     \\
Hamsterster      & 0.05     & 0.02       & 2.01     \\
CaAstroPh        & 0.56     & 0.20       & 7.65     \\
Epinions         & 2.16     & 0.58       & 5.21     \\
Amazon           & 1.67     & 2.66       & 1.88     \\
DBLP             & 1.74     & 5.16       & 10.30    \\
Youtube          & 10.07    & 11.67      & 13.70    \\
Skitter          & 33.52    & 28.03      & 77.37    \\
Wikipedia        & 144.06   & 108.16     & 239.03   \\
Livejournal      & 147.12   & 322.50     & 438.48   \\
Orkut            & 1008.08  & 983.59     & 3063.54  \\ \hline
\end{tabular}
\caption*{Time (seconds)}
}
\end{minipage}

    \caption{Left: Largest output clusters and their respective densities. Louvain and Greedy tend to output abnormally large clusters with poor density. For almost half of the networks, the largest set output by \mainproc{} is above 0.9, and even in the worst case, this density is above 0.2. Output sets are never more than a few hundred vertices in size. Right: Time taken by each of the well-performing methods. Typically, \mainproc{} is the fastest, and $k$-truss is the slowest.}
    \label{tab:largest}
    \vspace{-0.75cm}
\end{table*}

\mainproc{} gives significantly higher coverage at high densities than any of the competing methods. In \Tab{coverage}, we provide
the coverage for densities $0.5$ and $0.8$, across all the methods. For most instances, \mainproc{} gets significantly more coverage at high density. In \Tab{coverage}, we look at this restricted to sets with at least 5 vertices.
For example, in {\tt EUEmail}, coverage for Louvain and Nucleus are in single digits, whereas \mainproc{} coverage is 33\% at 0.5 and 23\% at 0.8. 
Similarly, for {\tt Livejournal}, we have 18\% and 16\% respectively, and for our largest dataset, the {\tt Orkut} social network, \mainproc{} gets 24\% and 16\%, higher than the other methods. For a comprehensive analysis, we include some datasets where other methods outperform \mainproc: typically, these are graphs with lower clustering coefficients. We lag in DBLP noticeably, but observe that many vertices are contained in smaller sets; a common feature of collaboration networks where a small group of authors frequently write with each other. For DBLP, \mainproc{} covers over 62\% vertices at both thresholds when including sets with fewer than 5 vertices, where as coverage is in the fifties for the others at 0.5, and Louvain drops to 27\% at 0.8. Greedy consistently performs the worst. 
\paragraph{Many sets of high density:} \mainproc{} outputs dense, reasonably sized sets with around tens to a few hundred vertices. 
Other methods give larger and sparse clusters with even half a million vertices. In \Fig{largest-scatter}, we plot the size and density
of the largest 20 sets output by \mainproc. For example, for the largest {\tt LiveJournal} and {\tt Orkut} datasets, we see all sets
with hundreds of vertices, with densities $0.3$ or higher. In contrast, the largest sets output by other methods typically have low density (in \Tab{largest}).  In general, both Greedy and Flowless output large sets of low edge density.
We do not give numbers of the nucleus decomposition, since it does not output a disjoint collection of sets. (The largest sets would simply
cover the entire graph, and most of the leaves have size three or four.) 
For Louvain, some outputs have half a million vertices and very poor density. In \Tab{summary}, we give the number of sets output
by \mainproc: they almost always have high density (\Tab{mean-density}).

\paragraph{Variation in coverage thresholds:} \mainproc{} coverage is high across different thresholds. Competing methods exhibit sharp drops in coverage at higher thresholds. A good example of this is Livejournal, where \mainproc{} has coverage of 25\% at 0.5 and 23\% at 0.8 (an 8\% drop in coverage); for Louvain, these are 13\% and 6\% (an 55\% drop in coverage) and for Nucleus, these are 21\% and 17\% (a 20\% drop in  coverage).  This trend holds across all datasets.

\begin{table}[!h]
    \centering
    \begin{tabular}{|p{3.5cm}|p{0.8cm}|p{2.5cm}|}
    \hline
    Title  & Year & Venue \\ \hline \hline
    \textbf{Top-k Spatial Joins} & 2005 & IEEE Transactions on Knowledge and Data Engineering \\ \hline
         Spatial hash-joins & 1996 & ACM SIGMOD Record \\ \hline
Partition based spatial-merge join & 1996 & ACM SIGMOD Record \\ \hline
Multiway spatial joins & 2001 & ACM Transactions on Database Systems\\ \hline
Efficient processing of spatial joins using R-trees & 1993 & International Conference on Management of Data \\ \hline
Spatial Joins Using R-trees: Breadth-First Traversal with Global Optimizations & 1997 & Very Large Data Bases \\ \hline
Spatial Join Indices & 1991 & International Conference on Data Engineering \\ \hline
Scalable Sweeping-Based Spatial Join & 1998 & Very Large Data Bases \\ \hline
Size separation spatial join & 1997 & International Conference on Management of Data \\ \hline
Slot Index Spatial Join & 2003 & IEEE Transactions on Knowledge and Data Engineering \\ \hline

Spatial join techniques & 2007 & ACM Transactions on Database Systems\\ \hline

    \end{tabular}
    \caption{An example of a cluster of 11 vertices from the DBLP citation network, with an edge density of 0.96. Despite having no knowledge of the labels, our algorithm is able to extract a dense network of papers on spatial joins from the network. The first row (in bold) is the first vertex.}
    \label{tab:dblp-d}
    \vspace{-0.2cm}
\end{table}
\paragraph{The largest set output:} We do further quantitative evaluation of the output in \Tab{largest}. Here, we look at the largest set output by each method, and the corresponding densities. Grouping vertices in absurdly large sets with hundreds and thousands of vertices fails to capture the essence of dense subgraph discovery. Even when we look at the largest sets output by \mainproc{}, it is at most a few hundred vertices. Moreover, densities are remarkably high: over 0.9 in five of the eleven networks we look at, and above 0.2 in all cases. However, Greedy, Flowless, and even Louvian assign vertices to much larger and sparser assets. This is especially notable in the larger datasets, where both Louvain and Greedy output sets as large as half a million vertices; Louvain being a bit better than Greedy in most cases. Owing to the nature of the hierarchical clusters of Nucleus, where most leaves are clusters of size 3 or 5, we do not do this comparison with that method.

Note that the greedy method optimizes for $D(H)=2E(H)/V(H)$ over all subgraphs $H\subseteq G$; thus hoping to maximize average degree. This implicitly penalizes edge density. In real networks, edge density is high only in small pockets of vertices that form dense subgraphs. Imagine a clique $K$ of size $50$; such a clique has 1225 edges, and the objective value $D(K)$ for the greedy algorithm for this clique is 49. For our goals, this is an excellent set with edge density 1, and we would like to recover. However, if $K$ is contained in a larger subgraph $L$ with 1k vertices and 50k induced edges, then $D(L)=50>D(K)$, but edge density is just 0.1, and will be favored by the greedy algorithm. Larger and larger subgraphs with poor edge density can have high average degree. Thus the classic greedy objective is a bad one if we want to maximize edge density.

\begin{table}
    \centering
    \begin{tabular}{|p{3.5cm}|p{0.8cm}|p{3.5cm}|}
    \hline
    Title  & Year & Venue \\ \hline \hline
    \textbf{Redesigning with traits: the Nile stream trait-based library} & 2007 & International Conference on Dynamic languages /International Smalltalk Joint Conference \\ \hline
Traits at work: The design of a new trait-based stream library & 2009 & Computer Languages, Systems and Structures \\ \hline
Automatic inheritance hierarchy restructuring and method refactoring & 1996 & ACM SIGPLAN conference on Object-oriented programming et al. \\ \hline
On automatic class insertion with overloading & 1996 & ACM SIGPLAN Notices \\ \hline
Back to the future: the story of Squeak, a practical Smalltalk written in itself & 1997 & ACM SIGPLAN conference on Object oriented programming et al. \\ \hline
Design of class hierarchies based on concept (Galois) lattices & 1998 & Theory and Practice of Object Systems \\ \hline
Refactoring class hierarchies with KABA & 2004 & ACM SIGPLAN conference on Object oriented programming et al. \\ \hline
Interfaces and specifications for the Smalltalk-80 collection classes & 1992 & Conference proceedings on Object oriented programming systems et al.\\ \hline
Identifying traits with formal concept analysis & 2005 & Automated Software Engineering \\ \hline
Traits: A mechanism for fine-grained reuse & 2006 & ACM Transactions on Programming Languages and Systems (TOPLAS) \\ \hline
    \end{tabular}
    \caption{An example of a cluster of 10 vertices from the DBLP citation network, with an edge density of 0.58. This cluster broadly focuses on object oriented programming, with a focus on traits. The first row (in bold) is the first vertex.}
    \label{tab:dblp-s}
     \vspace{-0.2cm}
\end{table}

\subsection{A case study}\label{sec:dblplab}
\begin{figure*}
     \centering
     \begin{subfigure}
         {\includegraphics[width= 0.24\textwidth]{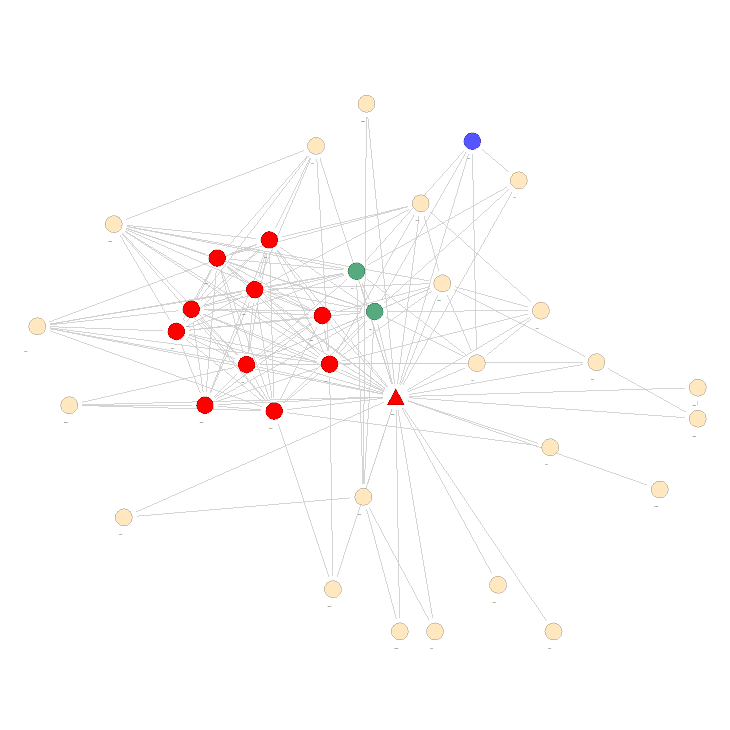}}
     \end{subfigure}
     \begin{subfigure}
         {\includegraphics[width= 0.24\textwidth]{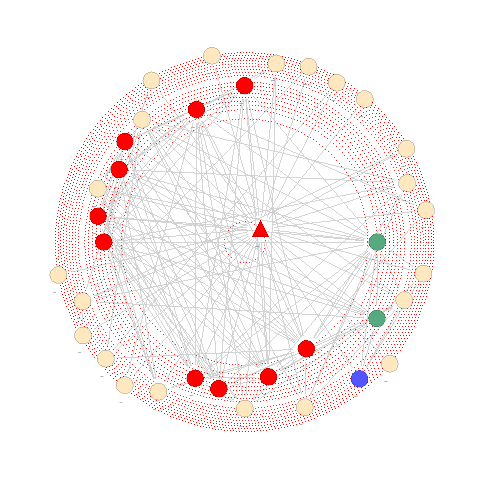}}
     \end{subfigure}
     \vline
     \begin{subfigure}
         {\includegraphics[width= 0.24\textwidth]{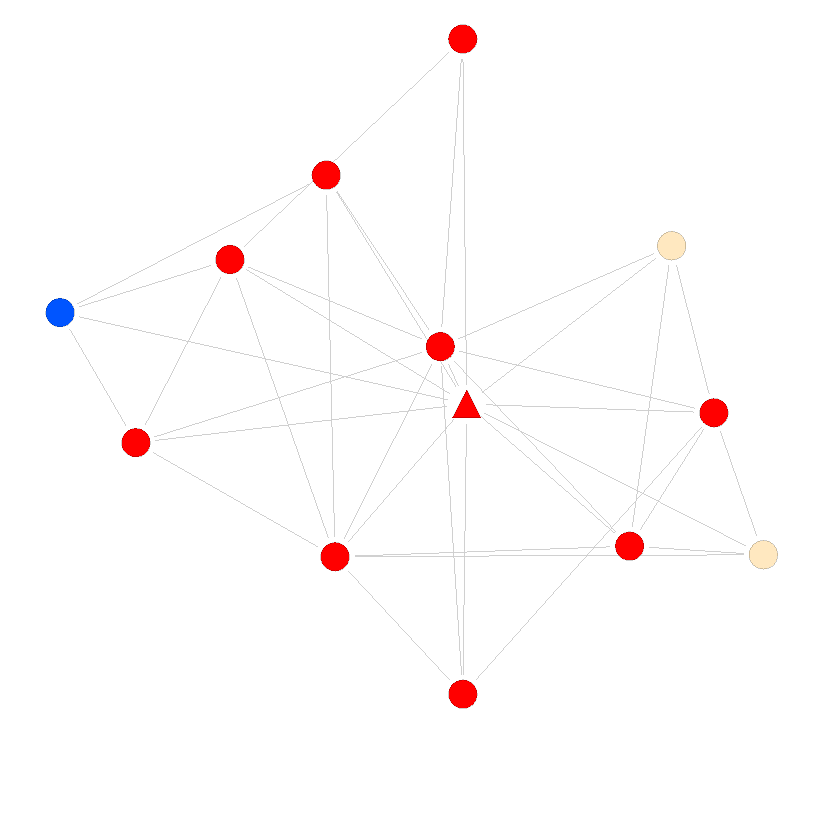}}
     \end{subfigure}
     \begin{subfigure}
         {\includegraphics[width= 0.24\textwidth]{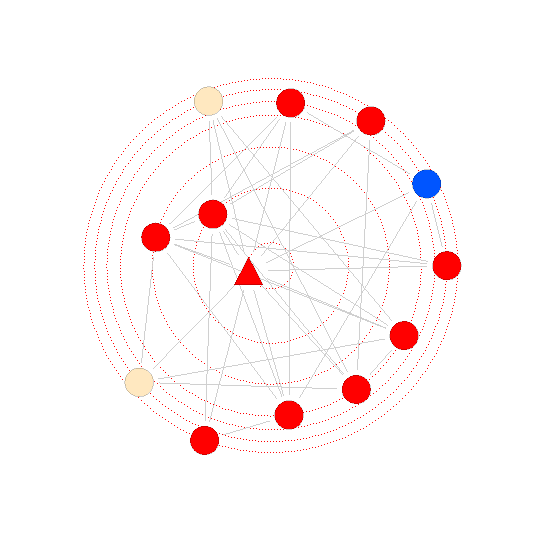}}
     \end{subfigure}
     \caption{On the left we have the subgraphs from \Tab{dblp-d}, and on the right from \Tab{dblp-s}. In both cases, we first draw a Fruchterman Reingold forced drawing of the subgraph, and then a radial illustration of vertices by closeness centrality. The lowest degree vertex is denoted by a triangle. In the first case, the lowest degree vertex induces a one hop neighborhood that is much larger than the cluster, and has many low degree vertices; the vast majority (20 of 35) of these are cleaned away as singletons, marked in pale yellow. The red vertices form our cluster. The blue and green vertices belong to some other non trivial clusters. In terms of closeness centrality, vertices in the cluster are much closer to the `central' (lowest degree) vertex. To the contrary, in the second cluster, the cluster is almost all of the one hop neighborhood, which has very few stray vertices. Only two of the vertices are peeled away as singletons.}\label{fig:subgraphs}
 \end{figure*}
 As an additional experiment, we look at a citation network of computer science papers taken from a bunch of sources including DBLP, ACM, and the Microsoft Academic Graph; hosted on AMiner~\cite{aminer}. We use version 2 from the source: this network has 660k vertices and 3M edges. Our algorithm extracted 15k non-trivial clusters. In \Tab{dblp-d}, we look at an example of such a cluster: it has 11 vertices, and an edge density of 0.96. Without any information about the semantic features of the network, our algorithm is able to extract a cluster of papers on spatial joins published in similar venues such as SIGMOD, TODS and VLDB in the rough span of a decade. As another example, let us pick a sparser cluster. For this, we pick a cluster of 10 vertices with an edge density of 0.58 in \Tab{dblp-s}. Here, we see a group of papers on object oriented programming focusing on traits, all from venues like TOPLAS and OOPSLA. 
 We get \emph{thousands} of such clusters, underscoring the potential for \mainproc{} for unsupervised knowledge discovery.
 
 As an illustrative example, let us examine how these sets are extracted from our network. Recall that \Alg{mainproc} always picks a vertex of minimal degree after passing \Step{clean}. In both cases, the extracted sets lie entirely in the one-hop neighborhood of this lowest-degree vertex. We direct the reader to \Fig{subgraphs} (on page 12), where we provide visual descriptions of these sets. We include two images for each set: a Fruchterman-Reingold drawing~\cite{FR91}, a method of forced graph drawing, and an arrangement of the same set radially by closeness centrality~\cite{GOLBECK201325}. Each image is of the immediate neighborhood of the first vertex in the network. The red vertices are the ones extracted by \mainproc{}, with the triangle being the lowest degree first vertex, and we look at the subgraph induced by this vertex and its immediate neighbors. 
 
 Our sets are formed of a tightly knit set of vertices around the lowest degree vertex: the rest are either cleaned away as singletons (marked in beige) or participate in other clusters (vertices marked in other colors). In the second case, while our subgraph is sparser, most of the neighboring vertices have a similar induced degree inside the neighborhood, and almost all of them are retained by our algorithm. In contrast, the subset in our first example is denser, but the one-hop neighborhood of the first vertex has vertices with significant variation in the induced degree. The vertices that form few (or no) triangles in the induced subgraph are all cleared away. In the closeness centrality figures, we note that the selected vertices are usually much closer to the first vertex. This is especially noticeable in the first set, where a lot of the other neighbors are peeled away as singletons.

 \paragraph{What is excluded from the sets?} The above examples show that everything \mainproc{} includes is relevant to the set as a whole. It is natural to ask how good \mainproc{} is at \emph{not} excluding relevant vertices. To examine this we look at the set of papers in the neighborhood of the first vertices that are not included in our output sets. We can confirm that this is indeed the case! This is especially true for the set in \Tab{dblp-d}: of the twenty-five neighbors of the first vertex excluded from the set, only four are on spatial joins. For the set in \Tab{dblp-s}, the papers are somewhat more similar (as exhibited in \Fig{subgraphs}, and they all concern object hierarchies. However, the venues are significantly different, and only one of them is about traits.

 \section{Conclusion and Future Directions}
 We present the algorithm \mainproc{} and show that it succeeds at efficiently extracting sets of densely connected vertices by leveraging triangle rich-sets. Our investigation opens several doors, which we propose as future research questions. 
 \begin{asparaitem}
    \item \emph{Cleaning: } The broader principle behind the cleaning procedure in \Step{clean} is the removal of `bad' edges with low participation in dense sets. For the extraction of dense sets, triangles are a good measure, and our proposed procedure does well. However, we believe that the idea of cleaning has rich applications as a sparsification procedure for a multitude of downstream tasks. To this end, a generalized cleaning procedure; or showing that this procedure works well for other tasks, would be a great direction for future research.
     \item \emph{Hierarchies: } The two well-performing algorithms we compare to, the Louvain algorithm and the Nucleus algorithm, are both hierarchical clustering algorithms. This leads to the question: is there a natural setup to organize the output sets of \mainproc{} to give a hierarchy of clusters?
     \end{asparaitem}
\begin{acks}
 Sabyasachi Basu, Daniel Paul-Pena and C. Seshadhri acknowledge the support of NSF DMS-2023495 and CCF-1839317.
\end{acks}

\bibliographystyle{ACM-Reference-Format}
\bibliography{RTR}

\end{document}

%% file: preamble.tex
\usepackage{a4,geometry}
\usepackage{enumitem}
\usepackage{graphicx}
\usepackage{amsmath,amsthm,mathtools}
\usepackage{paralist}
\usepackage{longtable}
\usepackage{colortbl}
\usepackage{hhline}
\usepackage{bm}
\usepackage{xspace}
\usepackage{url}
\usepackage{boxedminipage}
\usepackage{wrapfig}
\usepackage{subfigure}
\usepackage{ifthen}
\usepackage{color}
\usepackage{xcolor}
\usepackage{framed}
\usepackage{algorithm}
\usepackage{algpseudocode}
\usepackage{float}
\usepackage{rotating}
\usepackage{thmtools}
\usepackage{thm-restate}

\newcommand{\ignore}[1]{}

\newcommand{\cT}{\mathcal{T}}

\newcommand{\eps}{\varepsilon}

\newcommand{\poly}{\mathrm{poly}}

\newcommand{\Sec}[1]{\hyperref[sec:#1]{\S\ref*{sec:#1}}} 
\newcommand{\Eqn}[1]{\hyperref[eq:#1]{(\ref*{eq:#1})}} 
\newcommand{\Fig}[1]{\hyperref[fig:#1]{Fig.\,\ref*{fig:#1}}} 
\newcommand{\Tab}[1]{\hyperref[tab:#1]{Tab.\,\ref*{tab:#1}}} 
\newcommand{\Thm}[1]{\hyperref[thm:#1]{Theorem\,\ref*{thm:#1}}} 
\newcommand{\Fact}[1]{\hyperref[fact:#1]{Fact\,\ref*{fact:#1}}} 
\newcommand{\Lem}[1]{\hyperref[lem:#1]{Lemma\,\ref*{lem:#1}}} 
\newcommand{\Prop}[1]{\hyperref[prop:#1]{Prop.~\ref*{prop:#1}}} 
\newcommand{\Cor}[1]{\hyperref[cor:#1]{Corollary~\ref*{cor:#1}}} 
\newcommand{\Conj}[1]{\hyperref[conj:#1]{Conjecture~\ref*{conj:#1}}} 
\newcommand{\Def}[1]{\hyperref[def:#1]{Definition~\ref*{def:#1}}} 
\newcommand{\Alg}[1]{\hyperref[alg:#1]{Alg.~\ref*{alg:#1}}} 
\newcommand{\Clm}[1]{\hyperref[clm:#1]{Claim~\ref*{clm:#1}}} 
\newcommand{\Obs}[1]{\hyperref[obs:#1]{Observation~\ref*{obs:#1}}} 
\newcommand{\Rem}[1]{\hyperref[rem:#1]{Remark~\ref*{rem:#1}}} 
\newcommand{\Con}[1]{\hyperref[con:#1]{Construction~\ref*{con:#1}}} 
\newcommand{\Step}[1]{\hyperref[step:#1]{Step~\ref*{step:#1}}} 
\newcommand{\Assumption}[1]{\hyperref[assm:#1]{Assumption\,\ref*{assm:#1}}} 

\makeatletter

\newcommand{\ProbabilityRender}[2]{
  \@ifnextchar\bgroup%
  {\renderwithdist{#1}{#2}}
   {\singlervrender{#1}{#2}}
}
\newcommand{\singlervrender}[2]{%
   \ensuremath{\mathchoice
       {{#1}\left[ #2 \right]}
       {{#1}[ #2 ]}
       {{#1}[ #2 ]}
       {{#1}[ #2 ]}
   }
}
\newcommand{\renderwithdist}[3]{%
   \@ifnextchar\bgroup
   {\superfancyrender{#1}{#2}{#3}}
   {\ensuremath{\mathchoice
      {\underset{#2}{#1}\left[ #3 \right]}
      {{#1}_{#2}[ #3 ]}
      {{#1}_{#2}[ #3 ]}
      {{#1}_{#2}[ #3 ]}
     }
   }
}
\newcommand{\superfancyrender}[5]{
   \ensuremath{\mathchoice
      {\underset{#1}{{#1}}\left#4 #3 \right#5}
      {{#1}_{#2}#4 #3 #5}
      {{#1}_{#2}#4 #3 #5}
      {{#1}_{#2}#4 #3 #5}
   }
}

